\documentclass[10pt,conference,letterpaper]{IEEEtran}
\usepackage[T1]{fontenc}
\usepackage{amsmath,amssymb,amsthm,booktabs,array,microtype,balance}
\usepackage[hidelinks]{hyperref}
\usepackage{cite}
\usepackage{float}
\usepackage{graphicx}
\newtheorem{theorem}{Theorem}
\newtheorem{proposition}{Proposition}
\newcommand{\F}{\mathbb F_2}
\newcommand{\Z}{\mathbb Z}
\newcommand{\HX}{H_X}
\newcommand{\HZ}{H_Z}
\newcommand{\wt}{\operatorname{wt}}
\newcommand{\rk}{\operatorname{rank}_{\F}}
\newcommand{\row}{\operatorname{row}}

\newcommand{\supp}{\operatorname{supp}}
\newcommand{\code}[3]{[[#1,#2,#3]]}
\newcommand{\runex}[1]{\par\smallskip\noindent\textit{#1}\ }
\newcommand{\HXq}[1]{H_X^{(#1)}}
\newcommand{\HZq}[1]{H_Z^{(#1)}}
\usepackage{xcolor}

\usepackage{geometry}
\title{High-Girth QLDPC Codes from Circulant and Affine Permutation Matrices}
\author{\IEEEauthorblockN{Vahid Nourozi, Mohsen Moradi, Roxana Smarandache and David G. M. Mitchell}
}

\author{
Vahid~Nourozi$^\dag$, Roxana~Smarandache$^\ddag$,  %Roxana~Smarandache$^\ddag$, 
and David~G.~M.~Mitchell$^\dag$,
\\[2ex]
$^\dag$Klipsch School of Electrical and Computer Engineering,
New Mexico State University\\
Las Cruces, NM 88003, USA,
Email: nourozi@nmsu.edu; dgmm@nmsu.edu.\\
$^\ddag$Departments of Mathematics and Electrical Engineering, University of Notre Dame\\
Notre Dame, IN 46556, USA,
Email: rsmarand@nd.edu.\vspace{-1mm}
}

\begin{document}
\maketitle
\begin{abstract}
We study two explicit constructions of binary Calderbank--Shor--Steane quantum low-density parity-check (QLDPC) codes with check weight six. The first uses a fixed circulant shift pattern and yields parameters $\code{6N}{2N+2}{6}$ with Tanner graph girth twelve for $N=31$ and every $N\ge37$. A graph incidence interpretation gives the ranks and girths, and a uniform weight six logical operator establishes the exact distance. The second construction uses fixed matrix polynomials with circulant and affine permutation coefficients. Here, a polynomial commutation identity guarantees orthogonality for every cyclic lift. Four realizations give $\code{192}{16}{8}$, $\code{288}{16}{12}$, $\code{384}{16}{14}$, and $\code{576}{32}{20}$, with Tanner girth eight and computationally certified distances. We give reproducible construction data, distinguish classical cycle distance from quantum distance, and identify the limitations of both families.
\end{abstract}
%\begin{IEEEkeywords}
%Quantum LDPC codes, CSS codes, circulant permutation matrices, affine permutation matrices, Tanner girth.
%\end{IEEEkeywords}

\section{Introduction}
Binary Calderbank--Shor--Steane (CSS) codes are specified by two orthogonal parity-check matrices~\cite{calderbank,steane}. For quantum low-density parity-check (QLDPC) codes, sparsity must be achieved together with useful dimension and distance~\cite{mackay,breuckmann}. These requirements are distinct: orthogonality makes the stabilizer generators commute, whereas the distance depends on undetectable errors outside the stabilizer group. A large Tanner graph girth is desirable, but does not by itself determine quantum distance.

Permutation blocks provide a convenient description of structured check matrices. Quantum quasi-cyclic constructions~\cite{hagiwara,komoto}, pre-lifted classical protographs~\cite{prelift}, and algebraic girth conditions~\cite{girth, 9762914} explain how block permutation labels affect short cycles in the associated graph. For finite quantum codes, generalized bicycle and bivariate-bicycle constructions~\cite{panteleev,bravyi}, two-block group algebra codes~\cite{lin}, provide established frameworks in this domain. Affine permutations offer additional choices beyond circulant shifts that can be leveraged to improve the cycle distribution~\cite{kasai}. Related work by the authors and collaborators considers affine voltage lifts of bivariate-bicycle codes~\cite{nouroziapm, Nourozi2026sep4}, Q-learning searches over polynomial supports~\cite{nouroziq}, and polynomial-domain lifted-product constructions~\cite{nourozilp}. 

This paper studies two explicit constructions of binary Calderbank–Shor–Steane (CSS) QLDPC codes with check weight six. The main result is an explicit circulant-based family with
parameters $[[6N,2N+2,6]]$ for $N=31$ and every $N\ge37$. Orthogonality, rank, Tanner graph girth, and minimum distance are established analytically. In addition, Proposition~\ref{prop:girthbound} gives a restriction that does not depend on the circulant representation:
if every column of $H_X$ has weight two and $H_Z$
contains a nonzero row of weight $w$, then
$g_X\le2w$. Thus the first construction attains the
largest possible Tanner girth when both check matrices
have column weight two and row weight six.

The second construction uses a fixed commuting pair
of matrix polynomials with circulant and affine
permutation coefficients. Orthogonality and the
row and column weights hold for every outer lift size.
The ranks and distances reported for selected lift
sizes are finite computational results; they do not
establish an asymptotic distance growth law. Column weight two constructions with Tanner graph girth
twelve were previously obtained by
Komoto and Kasai~\cite{komoto}, and the mixed
construction admits a two-block group algebra
description~\cite{lin}. Our contribution is the
explicit specification and parameter analysis of
the displayed families, including a uniform
exact-distance proof for the first family.
We use $N=31$ and $q=2$ as worked examples.
\vspace{-0.2mm}
\section{Background}
\subsection{Permutation Blocks}
All matrix sums are over $\F$. For a permutation $f:\Z_N\to\Z_N$, define $P_f$ by
\begin{equation}
 \begin{gathered}
 [P_f]_{y,t}=1\ \Longleftrightarrow\ y=f(t),\\
 P_fP_g=P_{f\circ g},\qquad P_f^T=P_{f^{-1}}.
 \end{gathered}
 \label{eq:pm}
\end{equation}
A circulant permutation matrix (CPM) is $C_s=P_{t\mapsto t+s}$. An affine permutation matrix (APM) is $F_{a,b}=P_{t\mapsto at+b}$, where $\gcd(a,N)=1$. In particular, $C_0=I_N$ and
\begin{equation}
 C_sC_t^T=C_{s-t},\qquad F_{a,b}C_s=C_{as}F_{a,b}.
 \label{eq:affine}
\end{equation}
The second identity follows since both maps send $t$ to $at+as+b$. All shift subscripts and affine map values in
these definitions are reduced modulo $N$. These component permutation matrices are used to build the stabilizer matrices of a CSS code as described in the next subsection. 

\runex{Example 1:}
For $N=4$, $C_1$ sends $(0,1,2,3)$ to $(1,2,3,0)$, whereas $F_{3,1}$ sends them to $(1,0,3,2)$, i.e., 
\begin{equation}
 C_1=\begin{bmatrix}0&0&0&1\\1&0&0&0\\0&1&0&0\\0&0&1&0\end{bmatrix},\quad
 F_{3,1}=\begin{bmatrix}0&1&0&0\\1&0&0&0\\0&0&0&1\\0&0&1&0\end{bmatrix}.
 \label{eq:toy}
\end{equation}

\subsection{CSS Codes and Orthogonality Conditions}

We define the classical codes
\[
\begin{aligned}
C_X&=\{u\in\mathbb F_2^n\,|\,H_Xu^T=0\},\\
C_Z&=\{u\in\mathbb F_2^n\,|\,H_Zu^T=0\},
\end{aligned}
\]
where $u$ is a row vector and $H_X,\>H_Z$ are the parity-check matrices of appropriate size. Let $J_X,J_Z$ be the numbers of block rows and $L$ the common number of block columns in $H_X$ and $H_Z$, respectively, composed with one $N\times N$ permutation matrix in each block. The expanded matrices have $m_X=NJ_X$, $m_Z=NJ_Z$ rows and $n=NL$ columns. Thus $J_X$ counts blocks of checks, while $m_X$ counts individual binary checks. A qubit in block column $v$ and block position $t$ has index $vN+t$, $v\in\{0,1,\ldots,L-1\}$ and $t\in\{0,1,\ldots,N-1\}$. We denote the $N\times N$ matrix (a block) in location $(i,v)$ of the stabilizer matrices as $H_X(i,v),\>H_Z(i,v)\in \mathbb{F}_2^{N\times N}$, respectively, with corresponding scalar entries indexed as $[H_X]_{a,b}\in \mathbb{F}_2$ or from the (block) permutation matrix as $[H_X(i,v)]_{y,t}\in \mathbb{F}_2$, where the subscripts are over the appropriate range with indexing beginning at zero.

The CSS condition $H_XH_Z^T=0$ is equivalent to
$C_Z^\perp\subseteq C_X$, or equivalently
$C_X^\perp\subseteq C_Z$.
When this condition holds, let
$\mathcal Q=\mathcal Q(H_X,H_Z)$ denote the binary
CSS code whose X-type and Z-type stabilizer
generators are specified by the rows of
$H_X$ and $H_Z$, respectively~\cite{calderbank,steane}, we refer to these stabilizers as checks. For non-zero blocks $H_X(i,v)=P_{f_{iv}}$ and $H_Z(j,v)=P_{g_{jv}}$, set $S_{ij}=\{v\,|\,H_X(i,v)\ne0,\ H_Z(j,v)\ne0\}$. The CSS orthogonality condition can then be expressed as

\begin{equation}
 (\HX\HZ^T)(i,j)=\sum_{v\in S_{ij}}P_{f_{iv}\circ g_{jv}^{-1}}=0.
 \label{eq:cssblock}
\end{equation}
For each input coordinate, every output must occur an even number of times. 
A standard two-block CSS construction uses
square binary matrices $A$ and $B$ satisfying
$AB=BA$~\cite{panteleev,lin}. For square binary matrices, a useful special case is
\begin{equation}
 \begin{gathered}
 \HX=[A\mid B],\qquad \HZ=[B^T\mid A^T],\\
 \HX\HZ^T=AB+BA=0\quad\Longleftrightarrow\quad AB=BA.
 \end{gathered}\label{eq:bicycle}
\end{equation}
In this paper, the presented Construction I uses the blockwise condition \eqref{eq:cssblock}, while  Construction II constructs a commuting pair for \eqref{eq:bicycle}. In each case, a fixed permutation pattern supplies an algebraic orthogonality certificate before the lift size is chosen.

Let $r_X=\operatorname{rank}_{\mathbb F_2}H_X$
and $r_Z=\operatorname{rank}_{\mathbb F_2}H_Z$.
The CSS code $\mathcal Q(H_X,H_Z)$ has parameters
$[[n,k,d]]$, where $k$ is the number of encoded
qubits and $d_X,d_Z$ are its X-type and Z-type
logical distances. Explicitly,
\begin{align}
 k&=n-\rk\HX-\rk\HZ,\label{eq:k}\\
 d_X&=\min\{\wt(u)\,|\,u\in\ker\HZ\setminus\row(\HX)\},\nonumber\\
 d_Z&=\min\{\wt(u)\,|\,u\in\ker\HX\setminus\row(\HZ)\},\nonumber\\
 d&=\min(d_X,d_Z).\label{eq:distance}
\end{align}
The Tanner graph of each check matrix has a check vertex per row, a qubit vertex per column, and an edge per nonzero entry. Its girth is denoted by $g_X$ or $g_Z$; these are the girths of the two separate graphs.

\runex{Example 2: coordinates.}
For $N=31$, column $7$ of $C_3$ contains its
unique nonzero entry in row $10$.
When $C_3$ occupies block $(i,v)=(1,2)$ of $H_X$,
this entry has global row index
$31i+10=41$ and global column index
$31v+7=69$. Thus $[H_X]_{41,69}=1$ This convention is used throughout the first construction.

\section{A Circulant-Based Family with Girth Twelve}\label{sec:construction1}
\subsection{Case study: a fixed pattern and its incidence graphs}
Choose the integer lists
\begin{equation}
 \begin{split}
 a&=(0,1,3,8,12,18),\\
 b&=(0,-5,-5,0,12,12),\\
 c&=(0,0,2,8,2,8),
 \end{split}\label{eq:lists}
\end{equation}
and reduce them modulo $N$. Define
\begin{equation}
 \begin{split}
 \HX&=\begin{bmatrix}I&I&I&I&I&I\\C_0&C_1&C_3&C_8&C_{12}&C_{18}\end{bmatrix},\\
 \HZ&=\begin{bmatrix}C_0&C_{-5}&C_{-5}&C_0&C_{12}&C_{12}\\C_0&C_0&C_2&C_8&C_2&C_8\end{bmatrix}.
 \end{split}
 \label{eq:first}
\end{equation}
Here, $J_X=J_Z=2$ and $L=6$. Every block is non-zero, meaning that it contains one permutation matrix (a shift of zero represents the identity). Each matrix is $2N\times6N$, with row weight six and column weight two. 

 Each (qubit) column corresponds to an edge joining one check in the first block row to one in the second. Thus $\HX$ is the biadjacency matrix of a bipartite $(2,6)$-regular Tanner graph on $2N$ vertices. 
Let $G_X$ be the auxiliary (incidence) graph whose vertices correspond
to the rows of $H_X$. Since each column has two nonzero entries, this defines an edge joining the corresponding
two vertices in $G_X$. In the Tanner graph, this edge is a
two-edge path through the associated qubit vertex.
Consequently,
\[
g_X=2\,\operatorname{girth}(G_X).
\]
The same relation holds for $H_Z$ and its auxiliary graph

\begin{theorem}\label{th:first}
For $N=31$ or every integer $N\ge37$, \eqref{eq:first} defines a CSS code with parameters $\code{6N}{2N+2}{6}$, row weight six, column degrees two, and $g_X=g_Z=12$. Both check matrices have rank $2N-1$.
\end{theorem}
\begin{proof}
\emph{Orthogonality:} The four block products obtained via \eqref{eq:cssblock} correspond to block-row pairs $(0,0)$, $(0,1)$, $(1,0)$, and $(1,1)$. Using \eqref{eq:affine} we obtain the shift lists $-b$, $-c$, $a-b$, $a-c$ with components in each sum given as
\begin{equation}
\begin{array}{c|rrrrrr}
 -b&0&5&5&0&-12&-12\\
 -c&0&0&-2&-8&-2&-8\\
 a-b&0&6&8&8&0&6\\
 a-c&0&1&1&0&10&10
\end{array}.
\label{eq:difftable}
\end{equation}
Each row has even multiplicities, so all four sums vanish. For example, the third gives $C_0+C_6+C_8+C_8+C_0+C_6=0$. Orthogonality holds for every $N$.

\emph{Ranks:} The X incidence graph $G_X$ is connected because $a_1-a_0=1$. % there is a path x_{0,t}\to x_{1,t+1}\to x_{0,t+1}. Repeating it reaches every vertex in the first class, and the zero-shift edges connect these to every vertex in the second class.: two successive matching edges allow a unit translation within either vertex class. 
To see the rank directly, let $\lambda\in\mathbb{F}_2^{m_X}$ assign a binary coefficient to each check row. Requiring $\lambda^T\HX=0$ implies that the coefficients of $\lambda$ must be equal. This follow from the fact that the $G_X$ is connected and the column weight of $H_X$ is two. Thus the only solutions are $\lambda =0$ or $\lambda = 1$ %dependencies are the zero combination and the sum of all rows, 
and the rank is therefore $m_X-1 = 2N-1$. Multiplying block column $v$ of $\HZ$ by $C_{-b_v}$ modifies its row block shifts to $0$ and $c_v-b_v$, respectively. With
\begin{equation}
 \pi=(3,2,1,0,5,4),\quad (c-b)_{\pi(v)}=8-a_v,
 \label{eq:isomorphism}
\end{equation}
a column permutation, coordinate negation, and translation of one vertex class give a graph isomorphism. The same proof for $H_X$ follows from above and hence $\rk\HZ=2N-1$ as well. Finally,~\eqref{eq:k} gives $k=2N+2$.

\emph{Girth:}
The positive differences among the entries of $a$ are
$D=\{1,\ldots,12,15,17,18\}$.
The thirty elements of $D\cup(-D)$ are distinct
modulo $31$. They also remain distinct modulo
every $N\ge37$, since the difference between any
two has absolute value at most $36$.
Thus the $G_X$ has neither parallel
edges nor 4-cycles for these $N$ and $g_X=2G_X\geq 12$~\cite{prelift, girth}. Since any $2\times r$ array of circulants has girth bounded above by 12 \cite{sv12}, we have $g_X=12$. A similar proof holds for $G_Z$ and then $g_X, g_Z = 12$.
\emph{Distance:} %The result follows immediately from the fact that the Tanner graphs of $H_X$ and $H_Z$ have girth 12. Every cycle in a $2\times r$ array of circulants corresponds to a codeword of half that length \cite{sv12}. Thus $d_X = g_X/2 =6$ and $d_X=d_Z=d=6$. 
Every nonzero kernel vector of a graph incidence matrix contains a cycle, so \eqref{eq:distance} gives $d_X,d_Z\ge6$. For the reverse bound, let $e_t$ be the length $N$ unit column vector at $t\pmod N$, and set
\begin{equation}
 \begin{split}
 z_0&=e_1,\qquad z_1=z_2=z_3=0,\\
 z_4&=e_1+e_{-11}+e_{-5},\qquad z_5=e_{-11}+e_{-5}.
 \end{split}\label{eq:logical6}
\end{equation}
This vector has weight six for every allowed $N$. Its two X block syndromes vanish 
\begin{gather*}
 z_0+z_4+z_5=0,\\
 \begin{aligned}
 z_0+C_{12}z_4+C_{18}z_5
 &=e_1+(e_{13}+e_1+e_7)\\
 &\quad +(e_7+e_{13})=0.
 \end{aligned}
\end{gather*}
Every row of $\HZ$ has one nonzero entry in each of the six blocks. Any vector in $\row(\HZ)$ therefore has the same parity in all six blocks. The block parities of \eqref{eq:logical6} are $(1,0,0,0,1,0)$, so $z\notin\row(\HZ)$. Consequently $d_Z=d=6$.
\end{proof}

\subsection{Reusable conditions and the girth limitation} %{\bf 
The argument extends to other six-tuples $a,b,c\in\Z_N^6$. Sufficient conditions are: (i) all multiplicities in $-b,-c,a-b,a-c$ are even; (ii) there are a permutation $\pi$ of $\{0,\ldots,5\}$, $\epsilon\in\{\pm1\}$, and $\tau\in\Z_N$ such that
\begin{equation}
 (c-b)_{\pi(v)}=\epsilon a_v+\tau\quad(0\le v<6);
 \label{eq:designiso}
\end{equation}
and (iii) the thirty ordered differences $a_v-a_w$, $v\ne w$, are distinct, with $\gcd(N,a_1-a_0,\ldots,a_5-a_0)=1$. These conditions give, respectively, CSS orthogonality, isomorphic incidence graphs, and connected graphs without 4-cycles. The thirty differences are ordered: each of the six entries corresponding to the columns is subtracted from each of the other five, giving $6\cdot5=30$ values. They imply
\begin{equation}
 n=6N,\quad k=2N+2,\quad g_X=g_Z=12,\quad d=6.
 \label{eq:generalparams}
\end{equation}

\begin{proposition}\label{prop:girthbound}
If every column of $\HX$ has weight two and $\HZ$ has a nonzero row of weight $w$, CSS orthogonality implies $g_X\le2w$. The analogous statement holds with X and Z interchanged.
\end{proposition}
\begin{proof}
The rows in $H_Z$ are codewords in $\ker H_X$, and each codeword in a $2\times n$ matrix of weight $w$ is associated with a $2w$-cycle, hence $g_X\leq 2w$. 
\end{proof}
Hence replacing CPMs by other permutations cannot increase either girth above twelve while column degree two and check weight six are retained. Classical prelifting can avoid a circulant girth bound~\cite{prelift}; it does not remove this additional CSS bound. Conversely, a shortest cycle can be a stabilizer, so Proposition~\ref{prop:girthbound} does not upper-bound quantum distance.

\begin{table}[t]
\caption{Construction I from the fixed lists \eqref{eq:lists}. All distances are exact by Theorem~\ref{th:first}; $g_X=g_Z=12$, check weight six, and column degrees two.}
\label{tab:first}\centering
\begin{tabular}{rrrrr}
\toprule
$N$&$[[n,k,d]]$&$r_X=r_Z$&$R$&$\sigma$\\\midrule
31&$[[186,64,6]]$&61&0.3441&12.3871\\
37&$[[222,76,6]]$&73&0.3423&12.3243\\
40&$[[240,82,6]]$&79&0.3417&12.3000\\
41&$[[246,84,6]]$&81&0.3415&12.2927\\
43&$[[258,88,6]]$&85&0.3411&12.2791\\
47&$[[282,96,6]]$&93&0.3404&12.2553\\
63&$[[378,128,6]]$&125&0.3386&12.1905\\
128&$[[768,258,6]]$&255&0.3359&12.0938\\\bottomrule
\end{tabular}
\end{table}

Table~\ref{tab:first} lists realizations obtained by
changing only $N$. The rate is
\[
R=\frac{k}{n}=\frac13+\frac{1}{3N},
\]
and the minimum distance remains six.

\section{Mixed Circulant--Affine Construction}
\label{sec:mix}
We next give a fixed polynomial construction with column degree three.
Its coefficients combine circulant shifts with an affine permutation,
using the established permutation block framework~\cite{kasai}.
The modulus used is $N_0=24$, where $C_s$ represents $t\mapsto t+s$ and
$F=F_{7,7}$ represents $t\mapsto7t+7$ on $\mathbb{Z}_{N_0}$. These permutation matrices have size $N_0\times N_0$, all shifts are reduced modulo $N_0$, and $I=I_{N_0}$.
%All block subscripts are reduced modulo $24$ and $I=I_{24}$.
With a formal variable $z$ commuting with every coefficient, define
\begin{equation}
\begin{split}
A(z)&=\begin{bmatrix}I+zC_2&zF\\F&I+zC_2\end{bmatrix},\\
B(z)&=\begin{bmatrix}I&zC_8+z^3C_{10}\\C_8+z^2C_{22}&I\end{bmatrix}.
\end{split}
\label{eq:mix-base}
\end{equation}
These are $2 \times 2$ arrays whose entries are polynomials in $z$ with coefficients that are $24 \times 24$ permutation matrices. The integer $q \geq 2$ is the outer lift size. Let $S_q$ be the $q \times q$ permutation matrix of $u \mapsto u+1$ on $\mathbb{Z}_q$. Each qubit and check then has an outer coordinate $u \in \mathbb{Z}_q$, independent of the internal coordinate $t \in \mathbb{Z}_{24}$. The outer lift of a matrix polynomial $M(z)$ replaces each term $z^h M_h$ by $S_q^h \otimes M_h$ so that
\begin{equation}
\begin{gathered}
M_q=\sum_h S_q^h\otimes M_h\quad\text{if }M(z)=\sum_hz^hM_h,\\
\HXq{q} = [A_q \mid B_q], \qquad \HZq{q} = [B_q^T \mid A_q^T].
\end{gathered}
\label{eq:mix-specialize}
\end{equation}

Kronecker multiplication gives $(MM^\prime)_q = M_q M^\prime_q$. In this notation $zC_8$ is a coefficient with outer exponent one, and its lift is $S_q \otimes C_8$; the first factor shifts the outer coordinate and the second shifts the internal coordinate. For a matrix polynomial $M(z) = \sum_h z^h M_h$, define its adjoint
\[
M^\star(z) = M(z^{-1})^T = \sum_h z^{-h} M_h^T .
\]
The inverse variable accounts for the reversal of a cyclic shift under transposition. Since $S_q^T = S_q^{-1}$, we have $(S_q^h \otimes M_h)^T = S_q^{-h} \otimes M_h^T$, and hence $(M^\star)_q = M_q^T$ \cite{nourozilp}. The polynomial check arrays are
\begin{equation}\label{eq:mix-adjoint}
H_X = [A \mid B], \qquad H_Z = [B^\star \mid A^\star],
\end{equation}
and their lifts are the binary checks in \eqref{eq:mix-specialize}. Since $(M^\star)^\star = M$, we have $H_X H_Z^\star = AB + BA$, and lifting gives $\HXq{q} (\HZq{q})^T = A_q B_q + B_q A_q$. Thus the single polynomial identity $AB = BA$ guarantees CSS orthogonality for every outer lift size $q$.

\begin{theorem}\label{th:mix}
The matrices in~\eqref{eq:mix-base} satisfy $AB=BA$.
For every integer $q\ge2$, the checks in~\eqref{eq:mix-specialize}
define a length $96q$ CSS code. Each check matrix has $48q$ rows,
row weight six, and column weight three.
\end{theorem}
\begin{proof}
Set $D=I+zC_2$, $U=zC_8+z^3C_{10}$, and
$V=C_8+z^2C_{22}$, so
$$A=\begin{bmatrix}D&zF\\F&D\end{bmatrix} \textrm{ and
 }B=\begin{bmatrix}I&U\\V&I\end{bmatrix}.$$
The circulant coefficients commute, giving $DU=UD$ and $DV=VD$.
The affine identity $F_{a,b}C_s=C_{as}F_{a,b}$ gives
\begin{equation}
FC_8=C_8F,\quad FC_{22}=C_{10}F,\quad FC_{10}=C_{22}F.
\label{eq:mix-affine}
\end{equation}
For example, $7\cdot22\equiv10\pmod{24}$.
Consequently $zFV=UF$ and $FU=zVF$.
The two block products are
\begin{equation}
\begin{split}
AB&=\begin{bmatrix}D+zFV&DU+zF\\F+DV&FU+D\end{bmatrix},\\
BA&=\begin{bmatrix}D+UF&zF+UD\\VD+F&zVF+D\end{bmatrix}.
\end{split}
\label{eq:mix-products}
\end{equation}
The identities $DU = UD$, $DV = VD$, $zFV = UF$, and $FU = zVF$ make the corresponding blocks of (19) equal, so $AB = BA$. Since the lift preserves sums and products, lifting gives $\HXq{q} (\HZq{q})^T = A_q B_q + B_q A_q = (AB + BA)_q = 0$ over $\mathbb{F}_2$ for every $q \geq 2$. 
In $D_q = I_q \otimes I + S_q \otimes C_2$, the two terms have outer exponents $0$ and $1$, so for $q \geq 2$ they lie in different $24 \times 24$ blocks and no entries cancel. In $B_q$, the two terms of each off-diagonal block differ in outer
shift by two. They occupy distinct outer positions for $q>2$.
For $q=2$, the shifts coincide, but the internal shifts $8,10$
and $8,22$ are distinct, so no entries cancel.
Hence $A_q$ and $B_q$ each have row and column weight three.
Their concatenations have the claimed weights and dimensions.
\end{proof}

The off-diagonal blocks can also be obtained directly
from the commutation conditions.
With $V=C_8+z^2C_{22}$, define
\[
U=zFVF^{-1}=zC_8+z^3C_{10}.
\]
Then $UF=zFV$ by construction.
Since $F^2=C_8$ commutes with $V$, we also have
$FU=zVF$.
The diagonal polynomial $D=I+zC_2$ commutes with
$U$ and $V$ because their coefficients are circulants.
These identities provide algebraic conditions for
selecting the blocks; they do not by themselves
give a lower bound on the quantum distance.

\runex{Example 3: forming the $q=2$ checks.}
Here $S_2=\left(\begin{smallmatrix}0&1\\1&0\end{smallmatrix}\right)$
acts on two outer positions, and each $C_s$ acts on 24 internal
positions. Reducing $S_2^2=I_2$ and $S_2^3=S_2$ gives
\begin{align}
A_2={}&I_2\otimes\begin{bmatrix}I&0\\F&I\end{bmatrix}
 +S_2\otimes\begin{bmatrix}C_2&F\\0&C_2\end{bmatrix},\nonumber\\
B_2={}&I_2\otimes\begin{bmatrix}I&0\\C_8+C_{22}&I\end{bmatrix}
 +S_2\otimes\begin{bmatrix}0&C_8+C_{10}\\0&0\end{bmatrix}.
\label{eq:mix-q2}
\end{align}
Both are $96\times96$; substitution into~\eqref{eq:mix-specialize}
produces the complete $96\times192$ binary checks.

We index rows and columns from zero. A row or column of $A_2$ or $B_2$ is labeled by $(u, b, t)$ with index $48u + 24b + t$, where $u \in \mathbb{Z}_2$ is the outer coordinate, $b \in \{0,1\}$ is the block position, and $t \in \mathbb{Z}_{24}$ is the internal coordinate. Columns of the $B_2$ half of $\HXq{2}$ are offset by $96$. The term $zC_8$ in the upper-right block of $B$ maps $(0,1,0)$ to $(1,0,8)$: column $24$ maps to row $56$. Hence $[\HXq{2}]_{56,120} = 1$. 
For an affine entry, the lower-left $F$ in $A$ sends $(0,0,1)$ to $(0,1,14)$, giving $[\HXq{2}]_{38,1} = 1$. A complete check obtained by the same rule is
\begin{equation}
%\supp\bigl(\HX(2)_{56,:}\bigr)=\{6,31,56,120,142,152\}.
\operatorname{supp}\bigl([\HXq{2}]_{56,:}\bigr) = \{6, 31, 56, 120, 142, 152\}.
\label{eq:mix-row}
\end{equation}
This also exhibits the three entries from each half and total weight six.

%\balance
\runex{Example 4: dimension and exact distance.}
Binary elimination gives $\operatorname{rank}_{\mathbb{F}_2} \HXq{2} = \operatorname{rank}_{\mathbb{F}_2} \HZq{2} = 88$, hence $k = 192 - 88 - 88 = 16$. 
Let the row vector $w$ have support
\begin{equation}
\supp(w)=\{11,20,59,68,123,132,171,180\}.
\label{eq:mix-witness}
\end{equation}
Direct multiplication gives both $\HXq{2} w^T = 0$ and $\HZq{2} w^T = 0$, while binary elimination gives%Direct multiplication gives both $\HX(2)w^T=0$ and
%$\HZ(2)w^T=0$, while binary elimination gives
\begin{equation}
\operatorname{rank}_{\mathbb{F}_2} \begin{bmatrix} \HXq{2} \\ w \end{bmatrix} = \operatorname{rank}_{\mathbb{F}_2} \begin{bmatrix} \HZq{2} \\ w \end{bmatrix} = 89.
\label{eq:mix-augmented}
\end{equation}
Thus $w$ lies outside both stabilizer row spaces and supplies
weight eight logical operators of both types.
The matching lower bound uses the archived completed support
enumerations at weights four and six on both CSS sides.
Every kernel vector has even weight because every column has odd weight: $H x^T = 0$ implies $\mathbf{1}^T H x^T = \mathbf{1}^T x^T = 0$ over $\mathbb{F}_2$. 
A weight-two kernel vector would select two identical columns and therefore create a 4-cycle. The verified girth eight excludes this possibility.
Together these facts establish $d_X=d_Z=8$ and the
$\code{192}{16}{8}$ parameters.

\begin{table}[t]
\caption{Outer lifts of the fixed mixed base. All check matrices
have row weight six, column weight three, and Tanner girth eight.
Exact distances rely on the archived completed enumerations.}
\label{tab:mix-parameters}
\centering
\begin{tabular}{crrr}
\toprule
$q$&$[[n,k,d]]$&$\rk\HX=\rk\HZ$&$g_X=g_Z$\\\midrule
2&$\code{192}{16}{8}$&88&8\\
3&$\code{288}{16}{12}$&136&8\\
4&$\code{384}{16}{14}$&184&8\\
6&$\code{576}{32}{20}$&272&8\\\bottomrule
\end{tabular}
\end{table}
Table~\ref{tab:mix-parameters} changes only the outer
size $q$; all polynomial coefficients and exponents
remain fixed.
For each reported distance $d$, the archived
calculations exclude logical operators at the smaller
even weights and provide weight $d$ operators with
zero syndrome outside the stabilizer row space.

The performed search enumerates supports of bounded weight and, whenever a support has zero syndrome, tests whether it lies in the stabilizer row space. It suffices to consider supports with no proper nonempty zero syndrome subset: if a logical operator $S$ contains such a subset $T$, then $T$ and $S \setminus T$ both have zero syndrome and cannot both be stabilizers, so one of them is a lighter logical operator. Symmetry reduction uses only permutations that preserve both stabilizer row spaces. A search stage that did not finish gives no lower bound.

For context, the same construction has the established two-block
group algebra description~\cite{lin}.
On $G=\Z_{24}\rtimes_7\Z_{2q}$ use
$(t,v)(s,r)=(t+7^vs,v+r)$ and identify $v=2u+b$.
Then $A_q$ is identity plus left multiplication by $(7,1),(2,2)$,
while $B_q$ is identity plus right multiplication by $(22,5),(8,1)$.
Left and right multiplication commute.
This identifies the algebraic framework underlying the explicit blocks;
the mixed example is not a new general left/right construction.

\section{Numerical Results}
\label{sec:numerical_results}

In this section, we present some numerical results for  Construction I, described in Section~\ref{sec:construction1}, with $N=128$ and the
 shift lists (\ref{eq:lists}). The corresponding blocks are
\[
\begin{aligned}
H_X(0,v)&=I_N, & H_X(1,v)&=C_{a_v},\\
H_Z(0,v)&=C_{b_v}, & H_Z(1,v)&=C_{c_v},
\end{aligned}
\qquad v=0,\ldots,5,
\]
where all shifts are taken modulo $128$. These arrays
define binary matrices of size $256\times768$, each
with rank $255$, giving $k=768-255-255=258$.
The resulting code has parameters $[[768,258,6]]$.
The reference code is constructed from Definition~5
of Komoto and Kasai~\cite{komoto} with $L=6$ and $P=128$.
Both codes have parameters $[[768,258,6]]$, check
weight six, column weight two in each check matrix,
and Tanner graph girths $g_X=g_Z=12$.

Let $A_6^X$ and $A_6^Z$ denote the numbers of
X-type and Z-type logical operators of weight six,
respectively.
Distinct binary supports are counted separately,
without identifying operators that differ by a
stabilizer.
Our work gives $A_6^X=A_6^Z=3840$, whereas the
reference code gives $A_6^X=A_6^Z=4864$.
Thus the weight six multiplicity is reduced by
approximately $21.05\%$ on each CSS side.
This comparison concerns minimum weight logical
operators; no claim is made about the multiplicities
at higher weights or about a proportional reduction
in frame error rate (FER).

Fig.~\ref{fig:fer_768} compares the measured FERs under BP--OSD--10 decoding, reinforcement learning based
sequential decoding (RL--S), and reinforcement learning guided
multi-branch decoding (RLMB).
RL--S uses sequential BP updates~\cite{moradi_seq2026}
with an update schedule learned through reinforcement
learning~\cite{moradi_rl2026}.
The multi-branch decoder follows~\cite{Nourozi2026sep2}
and is denoted by RLMB in the figure.
Both RL decoders use $T=32$ maximum iterations.
At $p=10^{-3}$ and $p=5\times10^{-3}$, our work achieves
a lower measured FER under all three decoders. This improvement is obtained without changing the
length, rate, minimum distance, check weight, or girth.
The smaller number of minimum weight logical supports
provides a structural explanation consistent with the
observed gain: such operators give rise to errors with
identical syndromes but different logical effects. However, support multiplicity alone does not determine FER, and its $21.05\%$ reduction should not be interpreted as an equal reduction in decoding failures.

\begin{figure}[!t]
    \centering
    \includegraphics[width=\columnwidth]{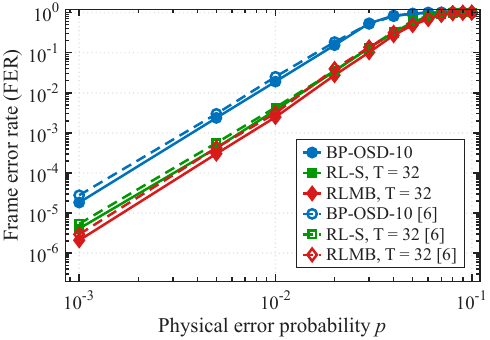}
    \caption{FER comparison between our work
    (solid curves) and the Komoto--Kasai
    instance~\cite{komoto} (dashed curves).
    Both codes have parameters $[[768,258,6]]$.
    RL--S and RLMB use $T=32$.}
    \label{fig:fer_768}\vspace{3mm}
\end{figure}

\section{Structural Limitations and Conclusion}
For check weight six and constant column degrees $j_X,j_Z$, let $\delta_X=m_X-\rk\HX$ and $\delta_Z=m_Z-\rk\HZ$ count independent row dependencies. Counting matrix entries gives
\begin{equation}
 R=1-\frac{j_X+j_Z}{6}+\frac{\delta_X+\delta_Z}{n}.
 \label{eq:rate}
\end{equation}
Construction I has $(j_X,j_Z)=(2,2)$ and one dependency on each side, giving $R=1/3+1/(3N)$. Construction II has $(j_X,j_Z)=(3,3)$, so its positive dimension comes entirely from row dependencies. For $q=2$, both matrices have 96 rows and rank 88; hence $\delta_X=\delta_Z=8$ and $R=16/192$.

The first family has constant distance six. Increasing $N$ increases the dimension and block length, but its relative distance tends to zero. Proposition~\ref{prop:girthbound} also shows why changing only the type of permutation cannot increase girth above twelve under the same column-degree and check-weight constraints. To establish a larger quantum distance for a different pattern, every shorter kernel vector must lie in the corresponding stabilizer space. A shortest cycle alone does not determine the quantum distance. For the mixed family, Theorem~\ref{th:mix} proves orthogonality and weights for every $q\ge2$. The listed ranks, girths, and distances are separate finite calculations. 

The two families illustrate different weight six designs. The circulant pattern has a uniform proof of rate and distance, while the mixed polynomial base gives larger certified distances at selected lengths. Their complete block descriptions specify the binary matrices without a separate orthogonality computation for each lift, and the worked examples show how the algebraic description determines individual checks, ranks, and logical operators.

\section*{Acknowledgement}
This material is based upon work supported by the National Science Foundation under Grant No. CCF-2145917. 
\clearpage
\balance
\bibliographystyle{IEEEtran}
\bibliography{bibliography}
\end{document}